\documentclass{article}

\usepackage{amsmath,amsfonts,amssymb,amsthm}
\usepackage{xspace}
\usepackage{latexsym}
\usepackage{graphicx}
\usepackage{epsfig}
\usepackage{array}
\usepackage{float}
\newtheorem{theo}{Theorem}
\newtheorem{prop}[theo]{Proposition}%[chapter]
\newtheorem{lem}[theo]{Lemma}%[chapter]
\newtheorem{defi}[theo]{Definition}%[chapter]

\usepackage{color}
\usepackage[usenames,dvipsnames,svgnames,table]{xcolor}

\newcommand{\mmod}[1]{\ensuremath{\,(\mathrm{mod}\ #1)\,}}

\newcommand{\set}[1]{\left\{#1\right\}}
\newcommand{\eps}{\varepsilon}
\newcommand{\lf}{\lfloor}
\newcommand{\rf}{\rfloor}

\newcommand{\ind}[1]{\mathbf{1}_{#1}}
\newcommand{\xn}{x_n}

\newcommand{\bin}{b^{(n)}}
\newcommand{\bi}{b}

\newcommand{\Ens}{\mathcal{L}}

\usepackage[latin1]{inputenc}
\author{Marie Albenque, Lucas Gerin}
\title{On the algebraic numbers computable by some generalized Ehrenfest urns}

\begin{document}
\maketitle
\begin{abstract}
This article deals with some stochastic population protocols, motivated by theoretical aspects of distributed computing.
We modelize the problem by a large urn of black and white balls from which at
every time unit a fixed number of balls are drawn
and their colors is changed according to the number of black balls among
them.The limiting behaviour of the composition of the urn when both 
the time and the number of balls tend to infinity is investigated and the proportion of
black balls is shown to converge to an algebraic number. We prove also that, surprisingly enough, not every algebraic number
can be ``computed'' this way.
\end{abstract}

\section{Introduction}
\subsection{Context and motivations}
The aim of this article is to tackle some questions of distributed computing in theoretical computer science,
from a statistical mechanics standpoint. Distributed computing deals with large computing systems
using many small processing elements. These small elements are thought as elementary
objects in a complex network whose interactions at a low level may be pretty difficult to understand
and modelize. There is a clear analogy with statistical mechanics, in which physical systems
are well described at a macroscopic level, while molecular-level phenomena seem chaotic.

This work is motivated by recent studies in \emph{population protocols}
(see \cite{Popu} for a detailed introduction), which are models
of decentralized networks consisting of mobile \emph{agents} interacting in pairs.
The way agents interact is known (and assumed to be simple) but not their movements.
These movements are driven by an ``adversary'', which picks, at each time step, two agents according
to a process only assumed to be \emph{fair} (roughly speaking, the fairness condition ensures that any possible configuration is eventually attained ; see again \cite{Popu} for a formal definition).

Let us be more precise. We are given a finite set $S$ of \emph{states}, a \emph{transition rule} $\phi:S^2\to S$, and $n\geq 2$ identical \emph{agents}, which
may be at any moment in one of the $\mathrm{card}(S)$ possible states.
A \emph{population protocol} associated to $\phi$ is a dynamical system $(\sigma_t)_{t\in\mathbb{N}}$ on
$S^n$ where at each time step two agents are chosen, and their states updated.
Updating is made according to $\phi$: if $x,y$, respectively in states $e,f$,
are chosen, then both their state is turned
into $\phi(e,f)$.
Population protocols are usually designed to compute predicates over a set of boolean variables.
A different but related question is addressed in this article: we use population protocols to compute real numbers instead of predicates.

%\subsection{Description of our model}

%Let $\mathcal{P}$ be a boolean function defined on the states of a population,
%\ie:
%$$
%\mathcal{P}:\set{e_1,\dots,e_q}^n\to \set{\mbox{true,false}}.
%$$
%It is a main question in distributed computing to ask wether the population
%protocol can compute the \emph{boolean expression} $\mathcal{P}$, meaning that
%for any fair scheme of updating,
%$$
%\left\{\mathcal{P}(\sigma_0)=\mbox{ true }\right\}
%\Leftrightarrow
%\left\{(\sigma_t)_{t_\in\mathbb{N}}\mbox{ converges to }e_1^n\right\}.
%$$
%In this setting the computational power of such systems is now fairly
%well understood as described in~\cite{Semi}.\\

\subsection{Model}
%When the updating scheme is
%random and the population is large, which real numbers can be computed by a population protocol?
To compute numbers with population protocols, we rely on the classical formalism of stochastic urn models, where each ball stands for an agent.

Let us describe more precisely the model we deal with. We fix an integer $k\geq 1$, a real number $x_0\in[0,1]$ and a \emph{rule}
$f:\ \set{0,\dots,k}\to\set{\mbox{black,white}}$.
At time 0, there are $n$ balls in the urn and each of them is randomly colored in
black with probability $x_0$ or in white with probability $1-x_0$, independently from the $n-1$ others.
The triplet $(k,f,x_0)$ is referred to as the \emph{rule} of the urn.

At each time unit, $k$ balls are picked randomly, uniformly
and independently from the past. Let $i$ be the number of black balls among
them, then all the $k$
balls are recolored in the color $f(i)$ and put back in the urn. This model is a generalization of the famous Ehrenfest
urn model (see for instance \cite{Dur}) which corresponds to
$k=1$ and $0\mapsto$ black $1\mapsto$ white (or equivalently at each time step a single ball is picked and its color is changed).

\begin{defi}\label{def:computable}
The real number $\alpha$ is said to be \emph{computable} if there exists a rule
$(k,f,x_0)$ such that, 
for any $\eps>0$, there exists $c>0$ such that 
$$
\mathbb{P}\left( \left|X_{\lf cn\rf}^{(n)} -\alpha\right |\geq\eps \right) \stackrel{n\to\infty}{\rightarrow} 0,
$$
where $X_\ell^{(n)}$ is the proportion of black balls at time $\ell$, and $\lf k\rf$ stands, as usual, for the largest integer smaller than $k$.
\end{defi}
Roughly speaking, it means that if we start with a large number $n$ of balls
and wait for a linear time in $n$ (larger than $cn$), the proportion of black balls is with high probability close to $\alpha$.

Before stating our main results, let us first discuss some important features of our model.
\begin{itemize}
\item The results we aim for are about the asymptotic properties of population protocols, when the size of the population grows to infinity.
From this perspective, the choice of picking agents uniformly at random appears to be the natural generalization of fairness to large (or even infinite) populations.
\item It is important to point out the strong assumption that the $k$ balls are all turned into the
\emph{same} color. From our original setting, this is motivated by the fact that in the complex network there is no hierarchy and not much communication between the agents: when they meet, they instantly all
take the same decision.
%\item We consider only the case of two possible states black and white. This comes from the fact that the number we compute corresponds to the proportion of balls in a given state. 
\item Because the number we compute corresponds to the proportion of black
balls, it is sufficient to consider only two possible states for the balls. This
is why we focus on this case.
\end{itemize}

%%%%%%%%%%%%%%%%%%%%%%%%%%%%%%%%%%%%%%%%%%%%%%%
\subsection{Description of the results}

The results of this paper are twofold.
We first describe in Section~\ref{Sec:Convergence} the asymptotic behavior of the proportion of black balls for any fixed rule. In Proposition~\ref{Prop:ConvODE}, the process $(X^{(n)}_{\lf nt\rf})_{t\geq 0}$ is shown to be close to the solution of an ordinary differential equation.
We then characterize in Theorem~\ref{Th:ConvPonctuelle} the numbers computable by a given rule as the roots of a polynomial related to this differential equation.

Section~\ref{Sec:Set} is devoted to studying the set of computable numbers. This relies on a combinatorial analysis of the latter polynomial. In Theorem~\ref{th:computablenumber}, we show that the set of computable numbers is dense in $[0,1]$. On the other hand, we prove in Proposition~\ref{Prop:ratnumber} that, surprisingly enough, rational numbers are not computable (except for a few explicitly exhibited ones).

We conclude this introduction by comparing our results with related works. In \cite{Urnes}, the black/white case with $k=2$ has already been handled but with
differences in the approach and the statements of the results.
The main difference is that the authors of \cite{Urnes} had to assume that the initial proportion $x_0$ of black balls is close to the number $\alpha$ one wants to compute\footnote{Indeed, the hypothesis "$|b_n(x)-b(x)|\to 0$" in (\cite{Urnes}, Th.2) implies that "$1-2(X_0^{(n)})^2\tfrac{n}{n-1}+X_0^{(n)}\tfrac{2}{n-1}\to 0$" and thus the initial proportion has to go to $\alpha$.} (as a counterpart, the main result of \cite{Urnes} gives a interesting and precise description of the fluctuations of $X^{(n)}_k$ around its mean, for $k$ large).
The techniques developed in the present paper allow us to free ourselves from this restrictive assumption.

Our main results (Theorem~\ref{th:computablenumber} and
Proposition~\ref{Prop:ratnumber}) may seem surprising as we might expect that
any algebraic number would be computable in our setting. This should be put in
perspective with a simultaneous and interesting result by Bournez, Fraigniaud and Koegler \cite{BournezFOCS}. They study stochastic population protocols in which agents are only picked in pairs but the number of possible colors is arbitrarily large and the new colors of the two balls may be different. 
In this different setting, any algebraic number is computable. Note that, as in \cite{Urnes}, the results in \cite{BournezFOCS} hold only with the assumption that the initial proportion of black balls is close to an equilibrium.

Let us also note that the link between the evolution of some stochastic population protocols and that of an associated ordinary differential equation has been used for the first time by Chatzigiannakis and Spirakis \cite{Dynamics} in a somewhat different context; they study some qualitative properties of the differential equation in order to discuss the stability of the underlying protocol.

Before going through details, we present now a simple introductory example.

\subsection{Heuristic : the example of $(3-\sqrt{5})/2$}\label{Sec:heuristic}
Take $k=2$ and consider the function $f: 0\mapsto \text{black } ; 1\mapsto
\text{white} ; 2\mapsto \text{white}$, as illustrated below
\begin{center}
\includegraphics[width=60mm]{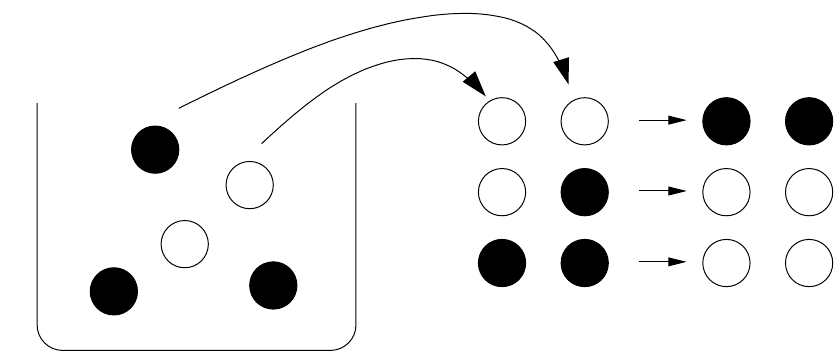}
\end{center}
Recall that $X_\ell^{(n)}$ denotes the proportion of black balls at time $\ell$. The sequence 
$(X_\ell^{(n)})$ defines a Markov chain on the set $\set{\tfrac{0}{n},\tfrac{1}{n},\dots,\tfrac{n}{n}}$, which admits a
unique invariant measure $\pi^{(n)}$.
Transition probabilities of the chain $X$ are clearly rational numbers so the components of $\pi^{(n)}$, as a solution
of a linear system of rational equations, are rational numbers. Its mean is thus rational.
Hence ergodic theorem for Markov chains states that almost surely:
$$
\frac{X^{(n)}_1+\dots +X^{(n)}_\ell}{\ell}\to p^{(n)}:=\mathrm{Mean}(\pi^{(n)})\in \mathbb{Q}.
$$
To get a hint for the asymptotic behavior of $p^{(n)}$, let us compute the conditional expectation of the increments of $X_\ell^{(n)}$:
%We do not pay attention to the exact expression of this mean, but run rather a non-rigorous
%computation that gives us a hint for its asymptotic behavior (when $n$ goes large) :
\begin{align}
\mathbb{E}\left[X_{\ell +1}^{(n)}-X_\ell^{(n)} | X^{(n)}_\ell=x\right]
= &+\frac2n \mathbb{P}\left(\mbox{both balls are white}\right)
-\frac1n \mathbb{P}\left(\mbox{one ball is white, one is black}\right)\notag\\
&-\frac2n \mathbb{P}\left(\mbox{both balls are black}\right),\notag\\
= &+\frac2n \frac{\binom{n-nx}{2}}{\binom{n}{2}}
-\frac1n \frac{nx(n-nx)}{\binom{n}{2}}
-\frac2n \frac{\binom{nx}{2}}{\binom{n}{2}},\notag\\
\stackrel{n\to\infty}{\sim } &\phantom{+}\frac1n \left(2(1-x)^2 -2x(1-x) -2x^2 \right).\label{Eq:Vanish}
\end{align}
Take now $\ell$ large, our system converges to its stationary regime, and thus
we expect the righ-hand term in \eqref{Eq:Vanish} to vanish. Hence, for large
$n$, $p^{(n)}$ should be close to the irrational number $(3-\sqrt{5})/2\approx 0.382\dots$, which is the only root of the polynomial $2(1-X)^2 -2X(1-X) -2X^2$ in $[0,1]$.
We let the balls ``compute'' $(3-\sqrt{5})/2$.

%--------------
\section{Limiting behavior of urns}\label{Sec:Convergence}

We study in this section the sequence $\mathbf{X}^{(n)}:=(X^{(n)}_\ell)_{\ell\geq 0}$  of the proportions of black balls in the urn. 
As mentioned above, this sequence is a Markov chain with state space $\set{\frac{0}{n},\frac{1}{n},\dots,\frac{n}{n}}$.
We denote by $E_f$ the set 
$$
E_f=\set{0\leq i\leq k;f(i)=\text{ black}},
$$
recall that the triplet $(k,f,x_0)$ (or equivalently, the triplet $(k,E_f,x_0)$) is the \emph{rule} of the urn.

Following the heuristic of Section~\ref{Sec:heuristic}, we associate to the rule $(k,f)$ the polynomial $b=b_f$ defined by
\begin{align*}
\bi(y)&=\sum_{i\in E_f} \binom{k}{i}(k-i)y^i(1-y)^{k-i}+\sum_{i\notin E_f} \binom{k}{i}(0-i)y^i(1-y)^{k-i}\\
&=\sum_{i\in E_f} \binom{k}{i}ky^i(1-y)^{k-i} -ky\\
&=k\mathbb{P}\left(\mathcal{B}_{k,y}\in E_f\right) -ky,
\end{align*}
where $\mathcal{B}_{k,y}$ is a binomial random variable with parameters $(k,y)$. In the above example ($k=2$ and $E_f=\set{0}$), this gives
$$
2\mathbb{P}\left(\mathcal{B}_{2,y}=0\right) -2y= 2(1-y)^2-2y,
$$
which is indeed equal to the polynomial in \eqref{Eq:Vanish}.

The meaning of $b(y)$ can be understood as follows:  when $n$ is large, picking $k$ balls uniformly among $n$ balls, a proportion $y$ of them
being black, almost amounts to perform $k$ times an experiment with a probability $y$ of success. The quantity
$\bi(y)$ then represents the expectation of the evolution of the number of black balls, after putting
back the $k$ recolored ones, as stated in the following lemma: 
\begin{lem}\label{Lem:ConvUnif}
For $y=c/n$ with $c\in \set{0,1,2,\dots,n}$, set
$$
\bin(y)=\mathbb{E}\left[X_1^{(n)}-X_0^{(n)} | X_0^{(n)}=y\right].
%\an(y)&=\mathbb{E}\left[\left(X_1-X_0\right)^2 | X_0=y\right].
$$
The map $y\mapsto n\bin(y)$ converges uniformly to $\bi(y)$ on the set
$\set{0,1/n,2/n,\dots,n/n}$. More precisely, for $n$ big enough,
$$
\max_{0\leq c\leq n} \left| n\bin(c/n) -\bi(c/n)\right| \leq 5k^3/\sqrt{n}.
$$
\end{lem}
\begin{proof}
The probability that $i$ balls are black when $k$ balls are picked in a urn that
contains $n$ balls among which $ny$ are black is equal to: 
$
\binom{ny}{i}\binom{n-ny}{k-i}\binom{n}{k}^{-1}, 
$
indeed to pick exactly $i$ black balls, we need to pick $i$ black balls among the
$ny$ black balls and $n-i$ white balls among the $n-ny$ white balls. 
For $y\in\set{0/n,n/n}$ there is nothing to prove since $n\bin(0)
-\bi(0)=n\bin(1) -\bi(1)=0$.
For any $y\in \set{1/n,2/n,\dots,(n-1)/n}$ we write
\begin{align*}
n\bin(y)-\bi(y)
&=\sum_{i=0}^k \left(k\ind{i\in
E_f}-i\right)
\left[\frac{\binom{ny}{i}\binom{n-ny}{k-i}}{\binom{n}{k}}-\binom{k}{i}y^i(1-y)^{k-i}\right]\\
&=\sum_{i=0}^k \left(k\ind{i\in
E_f}-i\right)\binom{k}{i}y^i(1-y)^{k-i}\left(\frac{\binom{ny}{i}\binom{n-ny}{k-i}}{\binom{k}{i}\binom{n}{k}y^i(1-y)^{k-i}}-1\right)
\end{align*}
Let us handle the last term :
\begin{equation}\label{Eq:TrucDegueu}
1- \frac{\binom{ny}{i}\binom{n-ny}{k-i}}{\binom{k}{i}\binom{n}{k}y^i(1-y)^{k-i}}
=
1- \frac{\binom{ny}{i}i!}{n^i y^i} \frac{\binom{n-ny}{k-i}(k-i)!}{n^{k-i}(1-y)^{k-i}} \frac{n^k}{k!\binom{n}{k}}.
\end{equation}
To show that \eqref{Eq:TrucDegueu} goes to zero (and hence, so does $n\bin(y)-\bi(y)$), recall that, when $j$
is fixed and $m$ goes to infinity, $\binom{m}{j}\sim e^{-j}m^j/j!$ : the last three terms in \eqref{Eq:TrucDegueu} converge to one.

In order to prove that the convergence is uniform, a little more work is needed. First, recall that a consequence of the Stirling formula
is that, for any integers $j\leq m$, 
\begin{equation}\label{Eq:Stirling}
\exp\big(-\frac{2j^2}{m}\big)\leq \frac{\binom m j}{m^j/j!}\leq 1.
\end{equation}
Plugging this in \eqref{Eq:TrucDegueu} gives that
$$
1-\exp(-2k^2/n)
\leq
1- \frac{\binom{ny}{i}\binom{n-ny}{k-i}}{\binom{k}{i}\binom{n}{k}y^i(1-y)^{k-i}}
\leq 
1-\exp(-\frac{2i^2}{ny})\exp(-\frac{2(k-i)^2}{n-ny}),
$$
and thus
$$
\left|1- \frac{\binom{ny}{i}\binom{n-ny}{k-i}}{\binom{k}{i}\binom{n}{k}y^i(1-y)^{k-i}}\right|\leq 
1- \exp\left(-2k^2 (\tfrac{1}{ny}+\tfrac{1}{n(1-y)})\right).
$$
If $1/\sqrt{n}\leq y \leq 1-1/\sqrt{n}$, then this last quantity is less than
${4}k^2/\sqrt{n}$. Then, for any $y\in[1/\sqrt{n};1-1/\sqrt{n}]$, one has
\begin{multline*}
|n\bin(y)-\bi(y)|\leq 
\sum_{i=0}^k \left|k\ind{i\in
E_f}-i\right|\binom{k}{i}y^i(1-y)^{k-i}\left|1- \frac{\binom{ny}{i}\binom{n-ny}{k-i}}{\binom{k}{i}\binom{n}{k}y^i(1-y)^{k-i}}\right|\\
\leq k \times 4k^2/\sqrt{n}\times \sum_{i=0}^k \binom{k}{i}y^i(1-y)^{k-i}=
{4 }k^3/\sqrt{n}.
\end{multline*}
We now deal with the case where $y< 1/\sqrt{n}$ and prove that both $n\bin(y)$
and $\bi(y)$ are close to $k\ind{0\in E_f}(1-y)^k$. We assume here that $n >
4k^2$ and therefore $y< 1/2k$ .
$$
n\bin(y)-k\ind{0\in E_f}(1-y)^k = k\ind{0\in E_f} \left(\frac{\binom{ny}{0}\binom{n-ny}{k-0}}{\binom{n}{k}}-(1-y)^k\right)
+\sum_{i=1}^k \left(k\ind{i\in E_f}-i\right) \frac{\binom{ny}{i}\binom{n-ny}{k-i}}{\binom{n}{k}}.
$$
Then, using again \eqref{Eq:Stirling} and inequalities $1-ky\leq (1-y)^k\leq
1-\tfrac{ky}{2}$ and $1-x\leq e^{-x}\leq 1-\tfrac{x}{2}$ for $0\leq x\leq 1/2$, 
\begin{align*}
|n\bin(y)-k\ind{0\in E_f}(1-y)^k |&\leq k\ind{0\in E_f} \Big|\frac{\binom{n-ny}{k}}{\binom{n}{k}}-(1-y)^k\Big|
+k \sum_{i=1}^k \frac{\binom{ny}{i}\binom{n-ny}{k-i}}{\binom{n}{k}}\\
&\leq k\ind{0\in E_f} \Big|\frac{\binom{n-ny}{k}}{\binom{n}{k}}-(1-y)^k\Big|
 + k\left(1-\frac{\binom{n-ny}{k}}{\binom{n}{k}} \right)\\
&\leq k(1-y)^k|\exp(-2k^2/(n-ny))-1| + k(1-(1-y)^{k}e^{-2k^2/(n-ny)})\\
&\leq k\times 1\times \frac{4k^2}{n} + k\left(1-(1-ky)(1-4k^2/n)\right)\\
&\leq 8k^3/n + k^2y - \frac{4yk^4}{n}\leq \frac{5k^2}{\sqrt n}.
\end{align*}

On the other hand,

\begin{align*}
|\bi(y)-k\ind{0\in E_f}(1-y)^k| &= |\sum_{i=1}^k \left(k\ind{i\in E_f}-i\right) \binom{k}{i}y^i(1-y)^{k-i}|\\
&\leq k(1-(1-y)^k)\leq k^2y\leq k^2/\sqrt{n}.
\end{align*}
This proves that for any $y< 1/\sqrt{n}$,
$$
|n\bin(y)-\bi(y)|\leq 5k^2/\sqrt{n}+k^2/\sqrt{n}\leq 6k^2/\sqrt{n},
$$
and concludes the case $y< 1/\sqrt{n}$. The case $y> 1-1/\sqrt{n}$ is symmetric.
\end{proof}

The rest of the section is devoted to the study of the convergence of
$\mathbf{X}^{(n)}$. For that purpose, we define $t\mapsto x(t)$ as the unique
maximal solution of the ordinary differential equation (ODE) such that
\begin{equation}\label{Eq:ODE}
 x'=\bi(x) \quad \text{and} \quad x(0)=x_0
\end{equation}
(recall that $x_0$ is the initial proportion of black balls in the urn).
First notice that since
\begin{equation}\label{Eq:TVI}
b(0)=k\mathbb{P}\left(\mathcal{B}_{k,0}\in E_f\right) = k\ind{0\in E_f} \geq 0 \quad \text{and} \quad
b(1)=k\mathbb{P}\left(\mathcal{B}_{k,k}k\in E_f\right) -k = k\ind{k\in E_f} -k\leq 0,
\end{equation}
this maximal solution $x_t$ actually remains in the interval $[0,1]$.
To describe the asymptotic behavior of the sequence $\mathbf{X}^{(n)}$, we speed
it up by a factor $n$, by setting $\xn(t)=X_{\lf nt\rf}^{(n)}$ and prove that
$x_n(t)$ is well approximated by $x(t)$ when $n$ is big: 
\begin{prop}\label{Prop:ConvODE}
For each rule and for any real numbers $t_0,\eps>0$, there exist $A,B >0$ such that, for each $n\geq 1$,
$$
\mathbb{P}\left(\sup_{t<t_0}|\xn(t)-x(t)|>\eps\right)\leq A e^{-B n}.
$$
\end{prop}
\begin{proof}
As this proposition can be seen as an instance of the general theory of large
deviations for Markov processes, sometimes known as Kurtz's Theorem (see \cite{Sch}), we only outline the
main ideas of the proof.

For sake of conciseness we set $X_k:=X^{(n)}_k$ and introduce the classical martingale
\begin{displaymath}
M_k=X_k-X_0-\sum_{\ell=0}^{k-1} \bin(X_\ell).
\end{displaymath}
This equation enables to rewrite $X_{\lf nt\rf}$ as:
$$
X_{\lf nt\rf}=X_0+M_{\lf nt\rf}+\int_0^{\lf nt\rf /n} nb^{(n)}(X_{\lf ns\rf})ds.
$$
Then
\begin{multline*}
X_{\lf nt\rf}-x_t=M_{\lf nt\rf}+\int_0^{\lf nt\rf /n} nb^{(n)}(X_{\lf
ns\rf})-b(X_{\lf ns\rf}) ds\\ + \int_0^{\lf nt\rf /n} b(X_{\lf
ns\rf})-b(x_s) ds + (X_0-x_0) + (x_{\lf nt\rf/n}-x_t),
\end{multline*}
since  $x_{\lf nt\rf/n}=x_0+ \int_0^{\lf nt\rf /n} b(x_s)ds$. 
Now, in order to bound,
$
f(t):=\sup_{s\leq t} |\xn(s)-x(s)|,
$
we write
\begin{multline*}
f(t)\leq \sup_{s\leq t}|M_{\lf ns\rf}|+\int_0^{\lf nt\rf /n}  |n\bin(\xn(s))-\bi(\xn(s))|ds
+\int_0^{\lf nt\rf /n}  |\bi(\xn(s))-\bi(x(s))|ds\\ + (X_0-x_0) +\sup_{s\leq
t}|x_{\lf ns\rf/n}-x_s|.
\end{multline*}
The probability for the first term to be large can be bounded with a
concentration inequality for martingales while the second term is bounded thanks
to Lemma \ref{Lem:ConvUnif}. The probability that the two last terms
are greater than $\eps$ is as small as desired.
Since the third term is smaller than $\sup |b'| \int f(s)ds$, an application of the Gr\"onwall Lemma gives a bound for
$\mathbb{P}(f(t)> \eps)$. Again, we
refer to (\cite{Sch}, p.76-84) or (\cite{EDO}, p.45-46) for details.
\end{proof}

The Cauchy-Lipschitz Theorem implies that $x'(t)$ is never equal to zero (unless $x$ is constant and equal to a root of $b$).
Hence any solution of \eqref{Eq:ODE} is monotonous, and converges to a root of $\bi$. If $\bi(x_0)\geq 0$ (resp. $<0$) then
the solution starting from $x_0$ converges to the smallest (resp. largest) root of $\bi$ greater
(resp. smaller) than $x_0$, denoted by $\alpha$. We gather this observation with
Proposition \ref{Prop:ConvODE} to obtain our main result:
%(recall that
%$X_\ell^{(n)}$ stands for the proportion of black balls at time $\ell$ in a urn of $n$ balls):
\begin{theo}\label{Th:ConvPonctuelle}
Assume that $\bi(x_0)\geq 0$ (resp. $<0$) and let $\alpha$ be the smallest (resp. largest) root of $\bi$ no smaller
(resp. no greater) than $x_0$. For any $\eps >0$, there exist some constants $c>0$ and $A,B>0$ such that for each $n$
$$
\mathbb{P}\left( \left|X_{\lf cn\rf}^{(n)} -\alpha\right |\geq\eps \right) \leq A e^{-B n}.
$$
In particular it implies that the rule $(k,E_f,x_0)$ \emph{computes} the number
$\alpha$, as defined in Definition~\ref{def:computable}.
\end{theo}

As regards a more quantitative aspect on time and space complexity, we point out the recent article \cite{AupyBournez} in which
this question is discussed for the case of $k=2$ and $\alpha=1/\sqrt{2}$ (but the method extends to other situations).

\begin{proof}[Proof of Theorem \ref{Th:ConvPonctuelle}]
First note that such an $\alpha$ always exists by~\eqref{Eq:TVI}.
Take now $c$ large enough, so that $|x(c) -\alpha|\leq \eps/2$. It suffices then to write
$$
\mathbb{P}\left( \left|X_{\lf cn\rf}^{(n)} -\alpha\right |\geq\eps \right)
\leq \mathbb{P}\left( \left|X_{\lf cn\rf}^{(n)} -x(c)\right |\geq\eps/2 \right).
$$
Proposition~\ref{Prop:ConvODE} gives the desired bound for the right-hand side.
\end{proof}

%---------------
\begin{figure}[h!]
\begin{center}
\includegraphics[width=11cm]{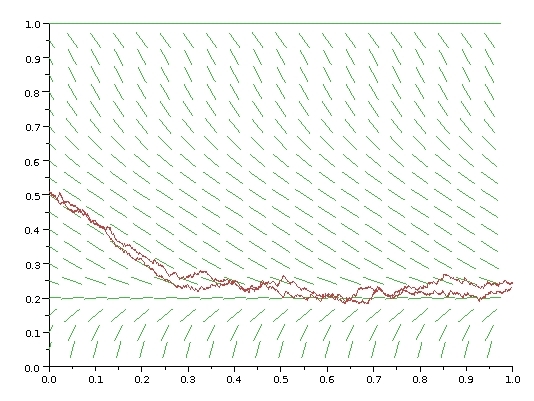}
\caption{Two simulations of $(X^{(n)})$, with $n=2000$ balls up to time $2000$,
with the flow of the corresponding ODE. Here, $k=8$, $E=\set{0,4,6,8}$ and  
$x_0=0.5$. The polynomial $b$ is equal to $8\left((1-x)^8+70x^4(1-x)^4+28x^6(1-x)^2+x^8\right)-8x$ and the corresponding
$\alpha$ is approximately equal to $0.2079$.}
\label{Fig:Simu}
\end{center}
\end{figure}
%--------------

%-----------------------------------------------------------------------------------
\section{The set of computable numbers}\label{Sec:Set}
We give in this section some properties about the set $\Ens$ of numbers that can be computed
by our urns. A first basic observation is that each
element of $\Ens$ is the root of a polynomial $b_f$ and hence is algebraic. Moreover we
have the following properties:
\begin{theo}\label{th:computablenumber}
The set $\Ens$
\vspace{-0.2cm}
\begin{enumerate}
\setlength{\itemsep}{1pt}
\item[\emph{(i)}] is symmetric with respect to $1/2$ ;
\item[\emph{(ii)}] is dense in $[0,1]$ ;
\item[\emph{(iii)}] contains numbers of any algebraic degree ;
\item[\emph{(iv)}] does not contain every algebraic number.
\end{enumerate}
\end{theo}
\begin{proof}
\noindent{\bf (i)} Let $\alpha$ in $\Ens$ and $(k,E,x_0)$ a be rule computing to $\alpha$.
We denote by $E^\star$ the set defined by:
$$
i\in E^\star\Leftrightarrow k-i\notin E,\quad .
$$
Let $b^\star$ be the polynomial associated to the new rule $(k,E^\star,1-x_0)$,
we have
\begin{align}\label{eq:polydual}
b^\star(1-\alpha) &=k\mathbb{P}\left(\mathcal{B}_{k,1-\alpha}\in E^\star \right)-k(1-\alpha)\\
&=-k\mathbb{P}\left(\mathcal{B}_{k,\alpha}\in E \right)+k\alpha.
\end{align}
Hence, $1-\alpha$ is a root of $b^\star$. One checks easily that if the solution of the ODE
$y'=b(y),y(0)=y_0$ converges to $\alpha$,
then the solution of $y'=b^\star(y),y(0)=1-y_0$ converges to $1-\alpha$.

\vspace{0.3cm}
\noindent{\bf (ii)} Let $a/b$ be a rational number in $[0,1]$, and $\eps,\delta$ two positive reals
such that
$$
(a/b -\eps,a/b +\eps) \subset (\delta, 1-\delta).
$$
We are looking for a number $\alpha\in (a/b \pm\eps)$ and a rule $(k,E,x_0)$ such
that the associated ODE converges to $\alpha$. In particular it is necessary
that:
\begin{equation}\label{Eq:Densite}
\mathbb{P}(\mathcal{B}_{k,\alpha}\in E)=\alpha.
\end{equation}
Fix for now the integer $k$, and consider the set
$$
E_{a,b}=\set{i\leq k; i\equiv 0,1,2,\dots,a-1 \mmod b}.
$$
The proof relies on the following lemma:
\begin{lem}\label{Lem:Binomial}
For any $0<\delta<1/2$, there exists $\lambda>0$ such that
for any integer $k$ and $x \in (\delta,1-\delta)$,
\begin{equation}\label{Eq:Binomial}
\left|\mathbb{P}\left(\mathcal{B}_{k,x}\equiv 0,1,\dots, a-1 \mmod{b}\right)- a/b\right|\leq e^{-\lambda k}.
\end{equation}
\end{lem}
\begin{proof}[Proof of Lemma \ref{Lem:Binomial}]
A proof based on linear algebra would give the best constant $\lambda$.
As we do not need here this exact value, we give a probabilistic and shorter proof.
The value modulo $(b-1)$ of a random variable $\mathcal{B}_{k,x}$ is the position at time $k$ of the walk
$\mathbf{X}=(X_\ell)_{\ell\geq 0}$ on
$\set{0,1,\dots,b-1}$ starting from $X_0=0$ and with probability transitions
$$
\mathbb{P}(X_{\ell+1}=X_\ell +1 \mmod b)=1-\mathbb{P}(X_{\ell+1}=X_\ell \mmod b)=x.
$$
starting from $X_0=0$. It is clear that this Markov chain admits as unique stationary measure the uniform measure
$\pi$ over $\set{0,1,\dots,b-1}$. By the general coupling inequality (see \cite{Lind} Chap.I.2.), the
desired quantity is smaller than
$$
\mathbb{P}\left(X_0\neq \tilde{X}_0,X_1\neq \tilde{X}_1, \dots ,X_k\neq \tilde{X}_k\right),
$$
where $X,\tilde{X}$ are two i.i.d. copies of $\mathbf{X}$, starting from $0$ and $\pi$. These two walks meet
necessarily if during $b$ successive steps $X$ goes $b$ steps forward while $\tilde{X}$ remains motionless.
This occurs with probability $x^b(1-x)^b$, hence
$$
\left|\mathbb{P}\left(\mathcal{B}_{k,x}\equiv 0,1,\dots, a-1 \mmod b\right)- a/b\right|\leq \left(1-x^b(1-x)^b\right)^{\lf k/b\rf},
$$
which decays exponentially in $k$, provided $x$ is bounded away from $0$ and $1$.
\end{proof}

Assume $k$ is a multiple of $b$, this ensures that $0\in E$ and $1\notin E$ and
thus that neither $0$ or $1$ is a root of $b$. So we might as well take a
smaller $\delta$ such that all the roots of $b$ in the interval $[0,1]$ belong
in fact to $(\delta,1-\delta)$. Let $k$ be such that $e^{-\lambda
k}<\eps$ and $x_0=0.5$. The solution of $y'=b(y)$ starting from $0.5$ converges to a root of $b$.
%$\alpha\mapsto \mathbb{P}(\mathcal{B}_{k,\alpha}\in E)=\alpha$.
By Lemma \ref{Lem:Binomial}, such a solution belongs to $(a/b \pm\eps)$ (see
Figure~\ref{Fig:plot}).
\vspace{0.3cm}
\begin{figure}
\begin{center}
\includegraphics[width=9cm]{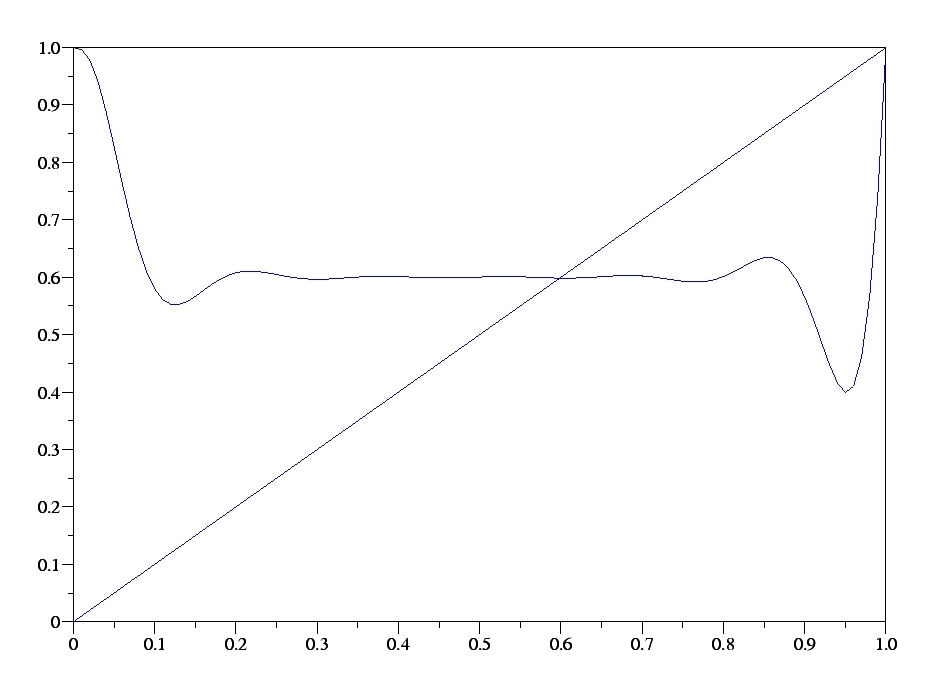}\\
\caption{A plot of the maps $x\mapsto x$ and $x\mapsto
\mathbb{P}(\mathcal{B}_{k,x}\in E_{a,b})$, for $k=30$,
$a/b=3/5$.\label{Fig:plot}}
\end{center}
\end{figure}

\noindent{\bf (iii)} Fix $k\geq 1$ and consider the set $E=\set{1}$.
The associated polynomial is
$$
k\alpha(1-\alpha)^{k-1}-\alpha.
$$
Its unique root in $(0,1)$ is
$$
x_0=1-\sqrt[k-1]{1/k},
$$
which has algebraic degree $k-1$.

\vspace{0.3cm}
\noindent{\bf (iv)} We give in fact in the next proposition a much stronger
result stating that almost no rational
numbers belong to the set $\Ens$.

\end{proof}

\begin{prop}\label{Prop:ratnumber}
Let $x = p/q$ be a rational number such that $\gcd(p,q)=1$ and $q\geq 4$ then
$x\notin \Ens$.
\end{prop}

Before proving this proposition, observe that the only rational numbers between
$0$ and $1$ that do not
satisfy the above conditions are $0$, $1$, $1/2$, $1/3$ and $2/3$. These numbers
all belong to $\Ens$ and are respectively computed by the rules $(1,\emptyset,0.5)$, 
$(1,\set{0,1},0.5)$, 
$(2,\set{1},0.5)$, $(3,\set{0,3},0.5)$ and $(3,\set{1,2},0.5)$.

We proceed by contradiction. Let $x=p/q$ such that
$\gcd(p,q)=1$ and $q\geq 4$ and assume that $p/q\in \Ens$. Since it implies that
$1-p/q \in \Ens$,  we can assume without
loss of generality that $p\geq 3$.
Let $(k,E,x_0)$ be one of the rules that admits
$p/q$ as a solution, we can hence write
\begin{equation}\label{eq:algebric}
\sum_{i\in E}\binom{k}{i}p^i(q-p)^{k-i} = pq^{k-1}.
\end{equation}

We now use some well-chosen reductions modulo $p$ to deduce from this relation
that $k\equiv 1 \mmod{p^n}$ for every $n$, which leads to a contradiction (take for
example $n$ equal to $k$). Reduction of \eqref{eq:algebric} modulo $p$ implies
that $\ind{0\in E}q^k \equiv 0 \mmod p$, which yields $0\notin E$, since
$\gcd(p,q)=1$. We go one step further, reducing \eqref{eq:algebric} modulo
$p^2$ and dividing by $p$ leads to the relation~:
\[
\ind{1\in E}k(q-p)^{k-1}\equiv q^{k-1} \mmod p,
\]
in which the left-hand side can be simplified into $\ind{1\in E}kq^{k-1}\mmod p$.
Since $\gcd (p,q)=1$, we obtain
$\ind{1\in E}k\equiv 1\mmod p$,
from which we readily deduce that $1\in E$ and $k\equiv 1 \mmod p$.

We now proceed by induction to show that $k\equiv 1 \mmod{p^n}$, for every $n\leq k$.
The following lemma will be useful:
\begin{lem}\label{lem:binomial}
Assume $k \equiv 1 \mmod{p^{n-1}}$, with $2\leq n\leq k$. Then for any $2< i\leq n$,
$\binom{k}{i}\equiv 0\ \mmod{p^{n-i+1}}$. Moreover if $p \not\equiv 2 \mmod 4$, the result is also true for $i=2$.
\end{lem}
\begin{proof}
It is enough to prove the lemma for $p$ being a power of a prime number,
otherwise writing the decomposition of $p=\prod p_i^{\alpha_i}$ into a product
of prime numbers and applying the result for each of the term gives the result.
We start with the classical relation:
\[
i(i-1)\binom{k}{i}=k(k-1)\binom{k-2}{i-2}.
\]
Since $k\equiv 1 \mmod{p^{n-1}}$, we get:
\begin{equation}\label{Eq:gcdmodulo}
\gcd(i,p^{n-1})\gcd(i-1,p^{n-1})\binom{k}{i} \equiv 0 \mmod{p^{n-1}}.
\end{equation}
Now, for $i\geq 4$, we have that $i\leq 2^{i-2}\leq p^{i-2}$, and hence $\gcd(i,p^{n-1})=\gcd(i,p^{i-2})$, since we assume that $p$ is a power of a
prime. Similarly $\gcd(i-1,p^{n-1})=\gcd(i-1,p^{i-2})$. Now, since $\gcd(i,i-1)=1$, 
$$
\gcd(i,p^{i-2})\gcd(i-1,p^{i-2})\leq p^{i-2}.
$$
Then \eqref{Eq:gcdmodulo} yields the desired result. 
The cases $i=2$ and $i=3$ are dealt with a direct computation.
\end{proof}

To continue the proof of Proposition~\ref{Prop:ratnumber}, we need to proceed differently depending on $p$ being or not equal to 2 modulo
4.
Assume first that $p\not\equiv 2 \mmod{4}$ and that we proved $k\equiv 1
\mmod{p^{n-1}}$ for some $n\geq 2$.
To carry on the recursion, we
write the reduction of \eqref{eq:algebric} modulo $p^{n+1}$,
%obtained by successive applications of Lemma~\ref{lem:binomial}:
$$
kp(q-p)^{k-1}+\sum_{i=2}^k\ind{i\in E}\binom{k}{i}p^{i}(q-p)^{k-i}\equiv
pq^{k-1} \mmod{p^{n+1}}.
$$
Lemma~\ref{lem:binomial} implies that each term in the sum of the l.h.s. vanishes. Expanding the remaining term $kp(q-p)^{k-1}$ and dividing both sides by $p$ gives
\begin{align*}
(1-k)q^{k-1} + \sum_{i= 1}^{k-1}(k-i)\binom{k}{i}p^iq^{k-1-i}&\equiv 0\
\mmod{p^n}\\
(1-k)q^{k-1} &\equiv 0 \mmod{p^n}.
\end{align*}
This proves by induction that $k\equiv 1\mmod{p^n}$ for any $n$, which leads to a contradiction 
and concludes the proof in this case.

Assume now that $p\equiv 2 \mmod{4}$ and $k\equiv 1 \mmod{p^{n-1}}$ for some $n\geq 2$.
%observe that in this case $k\equiv 1\ [p^{n-1}]$ implies $\binom{k}{2}\equiv 0
%[p^{n-1}/2]$. So instead
%of taking a reduction modulo $p^{n+1}$, we start by taking a reduction of
%Eq.\eref{eq:algebric} modulo $p^{n+1}/2$ to get:
%\[
%(1-k)q^{k-1}\equiv 0\ \big[\frac{p^n}{2}\big] \text{ and hence }k\equiv 1\
%\big[\frac{p^n}{2}\big].
%\]
We first observe that 
$$
q^k=(p+q-p)^k = \sum_{i\in E}\binom{k}{i}p^{i}(q-p)^{k-i} +\sum_{i\notin E}\binom{k}{i}p^{i}(q-p)^{k-i}.
$$
Thus \eqref{eq:algebric} can be written as
\begin{equation}
\sum_{i\notin E}\binom{k}{i}p^{i}(q-p)^{k-i} =q^k-pq^{k-1}\label{eq:algebdual}.
\end{equation}
Taking the reduction of the latter equation modulo $p^{n+1}$ and using that $\binom{k}{i}p^i\equiv 0\ \mmod{p^{n+1}}$ when $i>2$ gives
\begin{align}
(q-p)^k + \ind{2\notin E} \binom{k}{2}p^2(q-p)^{k-2}&\equiv q^k-pq^{k-1} \mmod{p^{n+1}}\notag\\
q^k-kq^{k-1}p+\binom{k}{2}q^{k-2}p^2 + \ind{2\notin E} \binom{k}{2}p^2(q-p)^{k-2}&\equiv q^k-pq^{k-1} \mmod{p^{n+1}}\notag\\
(1-k)pq^{k-1}+\binom{k}{2}q^{k-2}p^2 + \ind{2\notin E} \binom{k}{2}p^2(q-p)^{k-2}&\equiv 0 \mmod{p^{n+1}}\label{Eq:Fin?}.
\end{align}

We focus on the last term of the left-hand side and write:
\begin{align*}
\ind{2\notin E} \binom{k}{2}p^2(q-p)^{k-2}&=\ind{2\notin E} \binom{k}{2}p^2
\sum_{i=0}^{k-2}\binom{k-2}{i}p^i q^{k-2-i}\\
&=\ind{2\notin E}\binom{k}{2}p^2q^{k-2} +
\Big(k(k-1)p^2\Big)\ind{2\notin E}\frac{p}{2}\sum_{i=1}^{k-2}\binom{k-2}{i}
p^{i-1}q^{k-2-i}\\
&\equiv \ind{2\notin E}\binom{k}{2}p^2q^{k-2} \mmod{p^{n+1}},
\end{align*}
since $p/2 \in \mathbb{N}$ and $(k-1)p^2 \equiv 0 \mmod{p^{n+1}}$.

Then \eqref{Eq:Fin?} divided by $p$ can be written
$$
(1-k)q^{k-1}+(1+\ind{2\notin E})\binom{k}{2}q^{k-2}p \equiv 0 \mmod{p^{n}}.
$$
If $2\!\notin\! E$, the second term disappears and we are left with
$(1-k)q^{k-1}\equiv 0 \mmod{p^n}$ and hence $k\equiv 1 \mmod{p^{n}}$, otherwise we get
\[
(1-k)(q+\frac{pk}{2})\equiv 0 \mmod{p^n}.
\]
We conclude the proof by noticing that $(q+\frac{pk}{2})$ is both prime with
$p/2$ and odd, hence prime with $p$.

\section{Conclusion}

These first results raise some interesting theoretical questions in the current research on the computational power of population protocols.
Although our model is very general and allows to compute a large set of numbers, some algebraic numbers as "simple"
(on a computational point of view) as $1/5$ are not computable. A natural
question is then to ask if the set $\mathcal{L}$ has a nice structure : has it
interesting symmetries? can it be endowed with a certain algebraic structure
which is consistent with computability? In other words: does there exist an
operation $\otimes$ such that if $x$ and $y$ belongs to $\mathcal{L}$ then a
certain combination of their associated rules computes $x\otimes y$? 

As already mentioned, it is proved in \cite{BournezFOCS}
that any algebraic number is computable for $k=2$ and $q>2$ colors, but with the significant difference
that the $2$ balls may be turned into two different colors.
A question remains: what happens in our model with $k>2$ and two states if we consider more general rules for which the $k$
balls may be recolored differently from each other? With this new model it is possible to
compute any rational number but we still do not know if any algebraic number is
computable ; \cite{BournezFOCS} suggests that this should be the case.
It also would be interesting to study whether adding more colors but still requiring that the $k$ balls all turn into the same color has a bigger computational power. 

\vspace{4mm}

\subsection*{Acknowledgements}
We would like to thank O.Bournez and J.Cohen for some very interesting discussions
during the preparation of \cite{Urnes} that raised our interest for the subject and for showing us a preliminary version
of \cite{BournezFOCS}.

\nocite{*}
%\bibliographystyle{abbrvnat}
% use the following instead if you encounter problems 
\bibliographystyle{alpha}
\bibliography{AlbenqueGerin}

\end{document}